\documentclass[aps,prl,twocolumn,superscriptaddress,showpacs,10pt]{revtex4}

\usepackage[T1]{fontenc}
\usepackage{times}
\usepackage{color,graphicx}
\usepackage{subfig,array}
\usepackage{amsthm,amssymb,amsmath}

\newcommand\ket[1]{\ensuremath{|#1\rangle}}
\newcommand\bra[1]{\ensuremath{\langle#1|}}

\newcommand{\ketbra}[2]{| #1 \rangle\langle #2 |}

\newcommand\tr{\mathop{\rm tr}\nolimits}

\newcommand{\bes}{\begin{eqnarray*}}
	\newcommand{\ees}{\end{eqnarray*}} 
\newcommand{\bpm}{\begin{pmatrix}}
	\newcommand{\epm}{\end{pmatrix}}
\def\qed{\hfill $\Box$\medskip}
\def\IR{{\mathbb R}}
\def\IC{{\mathbb C}}

\def\cS{{\mathcal S}}
\def\cI{{\mathcal I}}
\def\cB{{\mathcal B}}

\def\cD{{\mathcal D}}

\def\cM{{\mathcal M}}
\def\cF{{\mathcal F}}

\def\bw{{\bf w}}
\def\bv{{\bf v}}

\def\bz{{\bf z}}

\def\diag{{\rm diag}\,}
\def\tr{{\rm tr}\,}

\def\conv{{\rm conv}\,}

\newtheorem{lemma}{Lemma}
\newtheorem{theorem}{Theorem}
\newtheorem{corollary}{Corollary}

\newtheorem{proposition}{Proposition}

\newtheorem{example}{Example}

\begin{document}


\title{Quantifying the coherence of pure quantum states}

\author{Jianxin Chen}%
\affiliation{Joint Center for Quantum Information and Computer Science, 
	University of Maryland, College Park, Maryland 20742, USA}%
\author{Shane Grogan}%
\affiliation{Department of Mathematics and Computer Science, Mount Allison University, Sackville, NB, Canada E4L 1E4}
\author{Nathaniel Johnston}
\affiliation{Department of Mathematics and Computer Science, Mount Allison University, Sackville, NB, Canada E4L 1E4}
\affiliation{Department of Mathematics and Statistics, University of Guelph, Guelph, ON, Canada N1G 2W1}
\author{Chi-Kwong Li}
\affiliation{Department of Mathematics, College of William and Mary,  Williamsburg, VA, USA  23187}
\author{Sarah Plosker}
\affiliation{Department of Mathematics \& Computer Science, Brandon University, Brandon,
	MB, Canada R7A 6A9}
	\affiliation{Department of Mathematics and Statistics, University of Guelph, Guelph, ON, Canada N1G 2W1}

\begin{abstract}
In recent years, several measures have been proposed for characterizing the coherence of a given quantum state. We derive several results that illuminate how these measures behave when restricted to pure states. Notably, we present an explicit characterization of the closest incoherent state to a given pure state under the trace distance measure of coherence.  We then use this result to show that the states maximizing the trace distance of coherence are exactly the maximally coherent states. We define the trace distance of entanglement and show that it coincides with the trace distance of coherence for pure states. Finally, we give an alternate proof to a recent result that the $\ell_1$ measure of coherence of a pure state is never smaller than its relative entropy of coherence. 
\end{abstract}

\date{\today}

\pacs{03.67.Ac, 03.65.Ta, 02.30.Mv, 03.67.Mn}

\maketitle

\section{I. Introduction}

One of the major goals in quantum information theory is to find effective ways of quantifying the amount of ``quantumness'' within a given system---that is, how much the system differs from any possible classical mechanical description of it. How this quantification is carried out varies heavily depending on context, however, as some quantum states might be useful for one quantum information processing task, yet useless for another.

When multiple quantum systems interact with each other, the resource of interest is typically \emph{entanglement}, the quantification of which has been well-investigated over the past two decades \cite{BDSW96,Bra05,HHT01,Hor01,Rai99,Shi95,Ste03,VT99}. However, when there is no interaction between different systems, the resource of interest is instead \emph{coherence}, or the amount that a state is in a superposition of a given set of mutually orthogonal states. With roots in quantum optics \cite{Glau63, Su63}, coherence is an essential operational resource in quantum information processing, and has been shown to be intimately related to entanglement \cite{YZCM15, YXGS15}; in fact, it has been shown that one can measure coherence via entanglement \cite{SSDBA15}.

Despite its usefulness, an effort to formalize the quantification of coherence has only begun somewhat more recently \cite{Abe06}. The defining properties of a \emph{proper} coherence measure were identified in \cite{BCP14}; for example, a state $\rho$ should have zero coherence under the proposed measure if and only if $\rho$ is \emph{incoherent} (i.e., it is diagonal in the pre-specified orthogonal basis, which we will always take to be the standard basis $\{\ket{i}\}_{i=1}^n$), since such states are exactly the ones that represent classical mixtures of the given basis states. We denote the set of all $n \times n$ matrices by $\cM_n$, the set of density matrices by $\cD_n$, and the set of incoherent states by $\mathcal{I}_n$, or simply $\cM$, $\cD$, and $\cI$ if the dimension is irrelevant or clear from context.

The two most widely-known coherence measures are the \emph{$\ell_1$-norm of coherence}, defined as the sum of the absolute values of the off-diagonal entries of the density matrix:
\[
C_{\ell_1}(\rho) := \sum_{i\neq j} |\rho_{ij}|,
\]
and the \emph{relative entropy of coherence} \cite{Abe06}:
\[
C_{\text{r}}(\rho) := S(\rho_{\operatorname{diag}}) - S(\rho),
\]
where $S(\cdot)$ is the von Neumann entropy and $\rho_{\operatorname{diag}}$ is the state obtained from $\rho$ by deleting all off-diagonal entries. Some other coherence measures that have been proposed recently include the \emph{trace distance of coherence} \cite{RPL15}, which is the trace norm distance to the closest incoherent state:
\[
C_{\tr}(\rho) := \min_{\delta\in \mathcal{I}}\|\rho-\delta\|_{\tr}=\min_{\delta\in \mathcal{I}}\sum_{i=1}^n|\lambda_i(\rho-\delta)|,
\]
where $\lambda_i(\rho-\delta)$ are the eigenvalues of the matrix $\rho-\delta$ and $\|\cdot\|_{\tr}$ is the trace norm, and the \emph{robustness of coherence} \cite{NBCPJA16}:
\[
C_{R}(\rho) := \min_{\tau \in \cD} \left\{ s \geq 0 \, \Big| \, \frac{\rho + s\tau}{1 + s} \in \mathcal{I}\right\}.
\]

The $\ell_1$-norm of coherence, relative entropy of coherence, and robustness of coherence have all been shown to be proper coherence measures, and it has been shown that the trace distance of coherence is a proper measure of coherence when restricted to qubit states or $X$~states \cite{RPL15}. Although the general case remains open, this partial result helps validate the fact that the trace distance is commonly used as a coherence measure. Additionally, simple formulas are known for all of these measures of coherence when restricted to pure states, except for the trace distance of coherence. Indeed, the $\ell_1$-norm of coherence and the relative entropy of coherence are \emph{defined} via explicit formulas, and the robustness of coherence of a pure state simply equals its $\ell_1$-norm of coherence \cite{NBCPJA16}. However, it was noted in \cite{RPL15} that it seems comparably difficult to compute the trace distance of coherence of a pure state, and evidence was given to suggest that a simple closed-form formula might not exist.

In this work, we investigate how these measures of coherence behave on pure states. 
In Section II, we use approximation theory to give characterizations of 
the best decoherent states for a given state with respect to the trace norm distance
and the spectral norm distance. One can use the results to check whether a decoherent
state is the best approximation for the given state in finitely many steps.
In Section III,  we give an ``almost formula'' for the trace distance of coherence of a pure state.
In particular, we show that it is given by one of $n$ different formulas (depending on the state), 
and which formula is the correct one can be determined simply by checking $\log_2(n)$ inequalities. 
Furthermore, one can construct the decoherent state nearest to the given pure state under the 
trace norm (and operator norm). We  also present examples and MATLAB code to demonstrate the 
efficacy of our method both analytically and numerically. 
In Section~IV, we  prove that the states maximizing the trace distance of coherence are 
exactly the maximally coherent states---another property that has already been known to 
hold for the other three measures of coherence.
In Section~V, we show that many measures of entanglement and coherence coincide, including the trace distance of coherence and the analogous trace distance of entanglement, 
which we define herein.
 In Section~VI, we give an alternate proof to a recent result \cite{RPL15} 
that says that the $\ell_1$ measure of coherence of a pure state is not smaller than its relative entropy 
of coherence. The result gives, as an immediate corollary,  an improvement to the known bound of the 
distillable entanglement of pure states in terms of their negativity. 
Finally, concluding remarks and open questions are discussed in Section~VII.

\section{II. Characterization of nearest decoherent states}
In this section we present some results
that allow us to give computable criterion to check whether 
a decoherent state $\rho$ is nearest to a given state $\rho$.
The results will also be used to construct the nearest coherent state
for a given pure state with respect to the trace distance by an efficient algorithm
(Theorem~\ref{thm:pure_algorithm}), and to prove that maximally coherent states $\rho$ yield 
the maximum value of $C_{\tr}(\rho)$ (Theorem~\ref{thm:maxcoherent}). 
We begin with a general result in approximation theory; for example, see \cite{Singer}.

\begin{proposition} \label{1.1}
	Suppose $W$ is a closed convex set 
	in a finite dimensional normed space $(V,|\!|\!|\cdot|\!|\!|)$, and $\bv\in V - W$.
	Then $\bw \in W$ is the best approximation of $\bv$ if and only if there is a linear functional 
	$f$ with $|\!|\!|f|\!|\!|^* \le 1$ such that $f(\bv-\bw) =|\!|\!|\bv-\bw|\!|\!|$ and $f(\bz) \ge 0$ for all $\bz \in V$
	such that $\bw-\bz \in W$.
\end{proposition}

If $|\!|\!|\cdot|\!|\!|$ is a norm on $\mathbb{C}^n$, the linear functional $f$ in the above proposition
has the form  
$f(X) = \tr(MX)$ for some $M$ in the dual norm ball of $|\!|\!|\cdot|\!|\!|$. It is well known that the norms $\|\cdot\|_{\tr}$ and $\|\cdot\|$ are dual to each other, and their respective  norm balls equal 
\bes
\cB_{\tr}
&=& \{ A \in \cM_n: \|A\|_{\tr} \le 1\} \\
&=& \conv\{ \pm \ketbra{u}{u}: 
\ket{u} \hbox{ is a unit vector in }  \IC^n\},\, 
\textnormal{ and}  \\
\cB
&=& \{ A \in \cM_n: \|A\|\le 1\} \\
&=& \conv\{ A\in \cM_n:  A \hbox{ is unitary} \}.
\ees
Here is another simple observation which is useful for our discussion.

\begin{lemma} \label{extreme}
Suppose $D = \diag(d_1, \dots, d_n) \in \cI$. Then 
\begin{eqnarray*}
\cF & = & \{F\in \cM_n: D - F \in \cI\} \\
&= & \{\diag(f_1, \dots, f_n): \sum_{j=1}^n f_j = 0, \ d_j \ge f_j, \ j = 1 \dots, n\}
\end{eqnarray*}
is a convex set. A matrix $F \in \cS$ is an extreme point if and only if 
there is at most one strict inequality among the inequalities $d_j \ge f_j$ for
$j = 1, \dots, n$. Consequently, there are at most $n$ extreme points for the set $\cF$.
\end{lemma}

\begin{proof}
Note that $\cF$ can be viewed as a compact convex set in $\IR^n$ consisting of vectors
$(f_1,\dots, f_n)$ governed by one 
equality $\sum_{j=1}^n f_j = 0$, and $n$ inequalities $d_j \ge f_j$ with $j = 1, \dots, n$.
In $\IR^n$, one requires $n$ linearly independent equations from the governing 
equalities  and inequalities  to determine an extreme point. The results follows.
\end{proof}

By Proposition~\ref{1.1} and the facts
about $\cB_{\tr}$ and $\cB$, we have the following result for pure states (unit vectors).

\begin{proposition} \label{2.1}
	Let $\ket{x}\in \IC^n$ be a unit vector
	such that $\ketbra{x}{x} \in \cD_n - \mathcal{I}$, and $\delta \in \mathcal{I}$.
	Then $\ketbra{x}{x}-\delta$ has exactly
	one positive eigenvalue $\lambda_1$ with a unique
	rank one eigenprojection $H = \ketbra{v}{v}$ such that 
	$$\|\ketbra{x}{x} - \delta  \|_{\tr} = 2\|\ketbra{x}{x}-\delta \| = 2\lambda_1.$$
	Let $D \in \cI$ and $\cF = \{F\in \cM_n: D - F \in \cI\}$.
	The following conditions are equivalent.
	\begin{itemize}
		\item[{\rm (a)}] 
		$\|\ketbra{x}{x}-D \|_{\tr} = \min\{ \|\ketbra{x}{x}-\delta\|_{\tr}: \delta \in \mathcal{I}\}$.
		\item[{\rm (b)}]
		$\|\ketbra{x}{x}-D \| = \min\{ \|\ketbra{x}{x}-\delta\|: \delta \in \mathcal{I}\}$.
		\item[{\rm (c)}] For every (extreme) element $F$ in $\cF$,  $\langle v|F| v\rangle \ge 0$.
	\end{itemize}
\end{proposition}

\begin{proof}
Note that if $\ketbra{x}{x}$ is not a diagonal matrix, then  
$\ketbra{x}{x} - \delta$ has eigenvalues
$\lambda_1 > 0 \ge \lambda_2 \ge \cdots \ge \lambda_n$ by Weyl's inequality.
Because $\tr(\ketbra{x}{x}-\delta) = \sum_{j=1}^n \lambda_j = 0$,  we have
$\|\ketbra{x}{x} - \delta \| = \lambda_1$ and $\|\ketbra{x}{x} -\delta \|_{\tr} 
= \lambda_1 - \sum_{j=2}^n \lambda_j = 2\lambda_1$
So, the first assertion, and the equivalence of (a) and (b) follow. 
In particular, the same matrix $D$ minimizes the trace norm and the operator norm.

A matrix $D \in \mathcal{I}$ is best approximation of $\ketbra{x}{x}$ with respect to the 
$\|\cdot\|_{\tr}$ if an only if there is an element $H$
in the dual norm ball of $\|\cdot\|$ such that
$\tr(\ketbra{x}{x}-D)H = \|\ketbra{x}{x}-D \|$ and $\tr(HF) \ge 0$ for any
$F$ such that $D -F \in \mathcal{I}$.
Suppose $H$ has spectral decomposition $ \sum_{j=1}^n \xi_j \ketbra{u_j}{u_j}$.
Then $\tr(\ketbra{x}{x}-D)H = \lambda_1$  can happen 
if and only if $H = \ketbra{v}{v}$.
Thus, conditions (a) and (c) are equivalent.
By standard results in convex analysis, 
$\langle v|F| v\rangle \ge 0$ 
for every element $F$ in $\cF$
if and only if 
$\langle v|F| v\rangle \ge 0$ 
for every extreme element $F$ in $\cF$.
\end{proof}

By Proposition \ref{2.1} (c) and Lemma \ref{extreme}, 
one can easily check whether a given $D \in \cI$ is nearest to a 
given pure state $\ketbra{x}{x}$ in finitely many steps. 

For mixed states, we have the following results.

\begin{proposition}
	Let $A \in \cD_n$, $D \in \cI$, and $\cF = \{F\in \cM_n: D - F \in \cI$.
	The following conditions are equivalent.
	\begin{itemize}
		\item[{\rm (a)}] 
		$\|A-D \|_{\tr} = \min\{ \|A-\delta\|_{\tr}: \delta\in \cI\}$.
		\item[{\rm (b)}]
		There is a Hermitian contraction $H \in \cM_n$ such that
		$\tr((A-D)H) = \|A-D\|_{\tr}$ and $\tr(FH) \ge 0$
		for every (extreme) element $F$ in $\cF$.
	\end{itemize}
\end{proposition}

\begin{proof}
By Proposition~\ref{1.1} and the remark after it, condition (a) holds
if and only if there is $H$ in the dual norm ball of the trace norm satisfying condition (b).
\end{proof}

We can obtain more information about the matrix $H$ in condition (b) of the above proposition 
using the spectral decomposition of $A-D
= \sum_{j=1}^p \mu_j \ketbra{u_j}{u_j} - \sum_{j=1}^q \nu_j \ketbra{v_j}{v_j}$,
where 
$$\mu_1, \dots, \mu_p, \nu_1, \dots, \nu_q > 0.$$ 
Then 
$$H = \sum_{j=1}^p\ketbra{u_j}{u_j} - \sum_{j=1}^q \ketbra{v_j}{v_j} + \sum_{j=1}^{n-p-q} \xi_j 
\ketbra{z_j}{z_j}$$
so that 
$\{\ket{u_1}, \dots, \ket{u_p}, \ket{v_1}, \dots, \ket{v_q}, \ket{z_1}, \dots, \ket{z_{n-p-q}}\}$ 
is an orthonormal basis for ${\mathbb C}^n$. 
Let $U$ be the unitary matrix whose columns are these basis vectors.
Then 
$$U^*(A-D)U = X_1 \oplus -X_2 \oplus 0_{n-p-q}$$ for some nonnegative diagonal 
matrices $X_1 \in \cM_p, X_2 \in \cM_q$, and $H$ will be of the form
$U^*(I_p\oplus -I_q \oplus X_3) U$
for some Hermitian contraction $X_3 \in \cM_{n-p-q}$. 

In particular,  if the best approximation 
element  $D\in {\mathcal I}$ of $A$ is such that 
$A - D$ is invertible, then $p+q = n$, and we have a Hermitian unitary 
$H = U^*(I_p \oplus -I_q)U$ satisfying the optimality condition.
Suppose $p+q < n$ and if $F_1, \dots, F_\ell$ are the extreme points of
$\cF$. Then we need to find a Hermitian contraction $X_3 \in \cM_{n-p-q}$
such that $-I_{n-p-q} \le X_3 \le I_{n-p-q}$ (in the positive semidefinite ordering) and
$$\tr U^*(I_p \oplus -I_q \oplus X_3)U F_j = \alpha_j + \tr X_3 \ge 0, \quad j = 1, \dots, \ell,$$  
where $\alpha_j = \tr(I_p \oplus -I_q \oplus 0_{n-p-q})U^*F_jU$ and
$G_j$ is the matrix obtained from $U^*F_jU$ by removing its first $p+q$ rows and columns.
One may check the existence of $X_3$ efficiently by positive semidefinite programming.

\medskip
Using a similar argument, we have the following.

\begin{proposition} \label{2.4}
	Let $A \in \cD_n$, $D \in \cI$ and $\cF = \{F: D - F \in \cF\}$.
	 Then the following are equivalent.
	\begin{itemize}
		\item[{\rm (a)}] 
		$\|A-D\| = \min\{ \|A-\delta\|: \delta \in \cI\}$.
		\item[{\rm (b)}] There is a Hermitian matrix $H \in \cM_n$ with $\|H\|_\tr = 1$ such that
		$|\tr((A-D)H)| = \|A-D \|$ and 
		$\tr(HF) \ge 0$ for every (extreme) element in $\cF$.
	\end{itemize}
\end{proposition}

\begin{proof}
By Proposition~\ref{1.1} and the remark after it, condition (a) holds
if and only if there is $P$ with $\|H\|_\tr = 1$ satisfying 
$|\tr((A-D)H)| = \|A-D \|$ and $\tr(HF) \ge 0$ for every (extreme) element in $\cF$.
\end{proof}

Again, one may get more information about the matrix $H$ in Proposition \ref{2.4} (b)
using the spectral decomposition of $A-D$. Suppose $\lambda_1$ and $\lambda_n$ are the largest 
and smallest eigenvalue of $A-D$ with eigenprojections $P_1$ and $P_2$, respectively.
If $\|A-D\| = \lambda_1 > |\lambda_n|$, then $H$ is a density matrix such that $P_1 - H$ is positive
semidefinite; if $\|A-D\| = |\lambda_n| > \lambda_1$, then
$-H$ is a density matrix such that $P_2 + H$ is positive
semidefinite; if $\|A-D\| = \lambda_1 = |\lambda_n|$, then 
$H = rQ_1 - (1-r)Q_2$ for some $r \in [0,1]$ and density matrices $Q_1, Q_2$ such that
$P_1 - Q_1, P_2-Q_2$ are positive semidefinite. 
Again, one can use positive semidefinite programming method to check the existence of 
$H$ satisfying $\tr FH \ge 0$ for the finite set of extreme points of the set $\cF$.

\section{III. The trace distance of coherence of a pure state}

We now present a characterization of $C_{\tr}(\ketbra{x}{x})$, where $\ket{x} \in \mathbb{C}^n$ is an arbitrary pure state. Note that there is a diagonal unitary $U$ and a permutation matrix $P$ such that $PU\ket{x}$ is a unit vector having non-negative entries 
$x_1 \ge \cdots \ge x_n\ge 0$ in descending order. 
We then have 
$$\|\ketbra{x}{x}-\delta \|_{\tr} = \|PU(\ketbra{x}{x} - \delta )U^*P^t\|_{\tr}$$
for any $\delta \in \cI$. 
So, we may replace $\ket{x}$ by $PU\ket{x}$. Without loss of generality, 
we will use this simplification to find the best approximation for 
$\ket{x} = (x_1, \dots, x_n)^t$ with $x_1 \ge \cdots \ge x_n\ge 0$, but we note that it straightforwardly applies to the general setting of an arbitrary unit vector in $\mathbb{C}^n$.

With this modification, we have the following.
\begin{theorem}\label{thm:pure_algorithm}
	Suppose $\ket{x} = (x_1, \dots, x_n)^t$ is a unit vector with entries $x_1 \ge \cdots \ge x_n \ge 0$.
	Let
	$s_\ell = \sum_{j=1}^\ell x_j$, $m_\ell = \sum_{j=\ell+1}^n x_j^2$, and $p_\ell = s_\ell^2-1-\ell m_\ell$ for $\ell \in \{1,\dots, n\}$.
	There is a maximum integer $k \in \{1, \dots, n\}$ satisfying 
	\begin{equation} \label{xkqk}
	x_k > q_k := \frac{1}{2ks_k}\left(p_k+\sqrt{p_k^2 + 4km_ks_k^2}\right).
	\end{equation}
	The unique best approximation of $\ketbra{x}{x}$ in $\cI$ with respect to the trace norm (and the operator norm) is $D  = \diag(d_1, \dots, d_k, 0, \dots, 0) \in \cI$ with
	$$d_j=\displaystyle \frac{x_j - q_k}{s_k - kq_k} \quad \hbox{ for } \quad 1\leq j \leq k.$$ 
	Furthermore,
	\bes
	C_{\tr}(\ketbra{x}{x}) = \|\ketbra{x}{x} - D \|_{\tr} &=& 2(q_k s_k + m_k),\\
\text{ and } \qquad 	\|\ketbra{x}{x} - D \| &=& q_k s_k + m_k.
	\ees
\end{theorem}

\noindent \emph{Proof:} \rm We may assume that $x_n > 0$, and use continuity for the general case.

First, we prove that there exists a matrix
$D  = \diag(d_1, \dots, d_k, 0, \dots, 0)$ such that $\ketbra{x}{x} - D$ has an eigenvector 
$\ket{v} = (q_k,\dots, q_k, x_{k+1}, \dots, x_n)^t$ corresponding to its largest eigenvalue (we will later show that this $D$ is the same one from the statement of the theorem). To this end, let 
$d_1, \dots, d_k, q, \mu > 0$ be variables satisfying the matrix equation 
$$(\ketbra{x}{x} - D)\ket{v} = \mu \ket{v} \ \hbox{ with } \   
\ket{v} = (q, \dots, q, x_{k+1}, \dots, x_n)^t.$$
Then  $\ketbra{x}{x}\ket{v} = D\ket{v} + \mu \ket{v}$. 
Because $\langle x|v\rangle = qs_k + m_k$, we have
\bes
&& (qs_k+m_k)(x_1, \dots, x_k, x_{k+1}, \dots, x_n)^t \\ & = &
(d_1q + \mu q, \dots, d_kq + \mu q, \mu x_{k+1}, \dots, \mu x_n)^t.
\ees
Summing up the first $k$ entries of the vectors on the left and right sides, we have
\begin{equation}\label{eq1}
(qs_k + m_k)s_k = k \mu q + q\sum_{j=1}^k d_j = k\mu q + q.
\end{equation}
Comparing the last $n-k$ entries of the vectors on both sides, we have
\begin{equation}\label{eq2}
qs_k + m_k = \mu.
\end{equation}
Substituting (\ref{eq2}) into (\ref{eq1}) to eliminate $\mu$, we have 
\begin{equation}\label{eq3}
f_k(q) := ks_k q^2 - q(s_k^2 -1-km_k) - s_km_k = 0.
\end{equation} 
Letting $q_k$ be the larger zero of $f_k(q)$, we have 
$$q_k = \frac{1}{2ks_k}\left(p_k + \sqrt{p_k^2+4km_ks_k^2}\right) > 0,$$
where $p_k = s_k^2-1-k m_k$. Note that 
$$q_1 = \left(\sqrt{1-x_1^2}+ x_1^2 - 1\right)/x_1 < x_1,$$
so there indeed exists a largest integer $k \in \{1, \dots, n\}$ such that
$x_k > q_k$. From this point forward, we fix $k$ at this largest possible value, and we note that $s_k = x_1 + \cdots + x_k \ge kx_k \ge kq_k$.
Define 
$$d_j := (x_j - q_k)/(s_k - kq_k) > 0 \quad \hbox{ for } j = 1, \dots, k.$$
By our construction, we have
$$(\ketbra{x}{x}-D) \ket{v} = \mu \ket{v}.$$
Furthermore, by (\ref{eq2}) we have
\bes
\|\ketbra{x}{x} - D \| & = & \mu = q_k s_k + m_k \qquad \hbox{ and } \\ 
\|\ketbra{x}{x} - D \|_{\tr} & = & 2\mu = 2(q_k s_k + m_k).
\ees
Next, we will prove that $q_k \geq x_{k+1}$ if $k < n$. To this end, let 
$f_k(q)$ be the polynomial defined by (\ref{eq3}). 
Then
\bes
&&f_{k+1}(x_{k+1})\\
&=&(k+1)s_{k+1}x_{k+1}^2
-x_{k+1}[s_{k+1}^2-1-({k+1})m_{k+1}]\\
&&-s_{k+1}m_{k+1}\\
&=& ks_kx_{k+1}^2+kx_{k+1}^3+s_{k+1}x_{k+1}^2 \\
&&- \, x_{k+1}[s_k^2+2x_{k+1}s_k-1-k(m_k-x_{k+1}^2)-m_{k+1}]\\
&&-(s_k+x_{k+1})(m_k-x_{k+1}^2)            \\
&=&ks_kx_{k+1}^2-x_{k+1}(s_k^2-1-km_k)-s_km_k\\
&& +s_{k+1}x_{k+1}^2-2s_{k+1}x_{k+1}^2
 + \, x_{k+1}(m_k-x_{k+1}^2)\\
&&-x_{k+1}m_k+s_{k+1}x_{k+1}^2+x_{k+1}^3\\
&=&ks_kx_{k+1}^2-x_{k+1}(s_k^2-1-km_k)-s_km_k\\
&=& f_{k}(x_{k+1}).
\ees

The product of the roots of the quadratic $f_k(q)$ 
equals $-s_km_k$, which is negative, so they have opposite signs.
As a result, for any positive number $\mu$, $f_k(\mu) \leq 0$ if and only if 
$\mu \leq q_k$. 
Since we chose $k$ so that $x_{k+1} \le q_{k+1}$ 
(recall that $k$ is the largest subscript so that $x_k > q_k$), we have $f_{k+1}(x_{k+1}) \le 0$. 
It follows that $f_{k}(x_{k+1}) = f_{k+1}(x_{k+1}) \le 0$ as well, i.e., $x_{k+1} \le q_k$ as desired.

\medskip
Finally, we will show that $D$ is the (unique) best approximation 
of $\ketbra{x}{x}$ in $\cI$ by establishing the following.

\medskip\noindent
{\bf Claim.} Let 
$\cF = \{F \in \cM_n: D-F \in \cI\}$.
Then $\bra{v}F\ket{v} \ge 0$ for any $F \in \cF$.

\medskip
First, note that if  $D = \diag(d_1, \dots, d_k, 0, \dots, 0)$ and if $F \in \cF$, then 
$f_{k+1}, \dots, f_n \le 0$. Hence
\begin{eqnarray*}
	\bra{v}F\ket{v} &=& \bra{v}\diag(f_1, \dots, f_n)\ket{v} \\
	&=& \sum_{j=1}^k f_j q_k^2 + \sum_{j=k+1}^n f_{j} x_j^2\\
	&=&  -\sum_{j=k+1}^n f_j q_k^2  + \sum_{j=k+1}^n f_j x_j^2 \\
	&=&  \sum_{j=k+1}^n f_j (x_j^2-q_k^2) \ \ge \  0,
\end{eqnarray*}
because we already showed that $q_k \ge x_{k+1} \ge \cdots \ge x_n$.
By Proposition~\ref{2.1}, $D$ is the best approximation element in 
$\cI_n$ of $\ketbra{x}{x}$ with respect to the operator norm and the trace norm. 
This completes the proof of the claim.

\medskip
To prove the uniqueness of $D$ and $k$, 
suppose $D_1$ is another element in $\cI$ such that
$\|\ketbra{x}{x}-D\| = \|\ketbra{x}{x}-D_1\|.$
Then
\begin{eqnarray*}
	\|\ketbra{x}{x}-D\| &=& \min_{\delta\in \cI} \|\ketbra{x}{x}-\delta\| \\
	&\le& \|\ketbra{x}{x} - (D+D_1)/2\| \\
	&\le& \|(\ketbra{x}{x} - D)/2\| + \|(\ketbra{x}{x} - D_1)/2\|.
\end{eqnarray*}
By \cite[Proposition 1.2]{CLH}, there are unitary matrices 
$V_1, V_2 \in \cM_n$ such that $V_1^\dag(\ketbra{x}{x}-D)V_2 = [\mu] \oplus Y$ and 
$$V_1^\dag(\ketbra{x}{x}-D_1)V_2 = [\mu] \oplus Z,$$
where  $Y, Z\in \cM_{n-1}$ are negative semidefinite matrices, and $\|\ketbra{x}{x} - D_1\| = \|\ketbra{x}{x}-D\| = \mu$  is the largest eigenvalue
of $\ketbra{x}{x} - D$ with eigenvector $\ket{v}$ as defined before.
Hence, if $\ket{u}$ is the first column of $V_2$ and $\ket{\tilde u}$ is the first column of $V_1$, then
$(\ketbra{x}{x}-D)\ket{u} = \mu \ket{\tilde u}$. 
It follows that  $\ket{u} = \xi \ket{v}$ for some $\xi \in \IC$ and
$\ket{\tilde u}  = \xi \ket{v}$.  Consequently, $(\ketbra{x}{x} - D_1) \ket{v} = \ket{v}$, and 
$D\ket{v} = D_1\ket{v}$ implying that $D = D_1$ as $\ket{v}$ has positive entries. This contradicts the 
assumption that $D \ne D_1$. 
By Proposition 2, we see that $D \in \cI$ attains
$\min_{\delta \in \cI}\|\ketbra{x}{x}-\delta\|$ if and only if
$D$ attains
$\min_{\delta \in \cI}\|\ketbra{x}{x}-\delta\|_{\tr}$. Thus, $D$ is the unique element in $\cI$ attaining 
$C_{\tr}(\ketbra{x}{x})$.

Because $k$ is the rank of the unique best approximation of $D$ in $\cI$ (with respect to the operator norm), we see 
that $k$ is unique, which completes the proof of the theorem. (Alternatively, if there is another $\tilde k$ satisfying (\ref{xkqk}), 
then one can use the construction in our proof to get $\tilde D$ of rank $\tilde k$ 
that best approximates $\ketbra{x}{x}$, which is a contradiction.) \qed 

Before proceeding, we note that the $k = 1$ and $k = n$ cases of Theorem~\ref{thm:pure_algorithm} actually simplify significantly:
\begin{corollary}
	Using the notation of Theorem \ref{thm:pure_algorithm}, we have the following.
	
	\begin{enumerate}
		\item The best incoherent approximation of $\ketbra{x}{x}$ is a rank one matrix, which must
		 equal $\diag(1,0,\dots,0)$, if and only if $x_1m_2 \ge 2x_2m_1$.
	
		\item The best incoherent approximation of $\ketbra{x}{x}$ is an invertible matrix, which must equal $D = \diag(d_1, \dots, d_n) \in \cI$ with 
		$$\quad d_j = \frac{1}{n} [ 1 - s_n(s_n-nx_j)] > 0 \ \ 
		\text{ for } \ \ j = 1, \dots, n,$$ if and only if $1 > s_n(s_n-nx_n)$.
	\end{enumerate}
\end{corollary}

\begin{proof}
	To prove statement~1, we note that $k = 1$ if and only if $x_2 \leq q_2$. This is equivalent to
	$0 \ge f_2(x_2)$, where $f_2(q)$ is the quadratic defined in (\ref{eq3}), as shown in the 
	proof of Theorem \ref{thm:pure_algorithm}. Explicitly, we have
	\begin{eqnarray*}
		0 &\ge& 2s_2x_2^2 - x_2(s_2^2 - 1 -2m_2) - s_2m_2  \\
		&=& 2(x_1+x_2)x_2^2 - x_2[(x_1+x_2)^2-1] + 2x_2m_2 - s_2m_2 \\
		&=& x_2[2(x_1+x_2)x_2 - (x_1^2+x_2^2+2x_1x_2 - 1)] \\ &&
		 + x_2 m_2 - x_1 m_2 \\
		&=& x_2[2x_2^2  + (1 - x_1^2 - x_2^2)] + x_2(1-x_1^2-x_2^2) - x_1m_2 \\
		&=& 2x_2m_1 - x_1m_2.
	\end{eqnarray*}
	
	To prove statement~2, note that $q_n=\frac1{ns_n}(s_n^2-1)$,  
	and the stated inequality is equivalent to $d_j >  0$ 
	for all $j$, which is to say  that $D$ is positive definite. 
\end{proof}

Although Theorem~\ref{thm:pure_algorithm} appears somewhat technical at first glance, it is very simple to use both numerically and analytically. On the numerical side, it provides an extremely fast algorithm for computing $C_{\tr}(\ketbra{x}{x})$. Although it might seem somewhat time-consuming at first to find the value of $k$ described by the theorem, the proof of the theorem showed that if $q_k < x_k$ then $q_j < x_j$ for all $j < k$. Thus we can search for $k$ via binary search, which requires only $\log_2(n)$ steps, rather than searching through all $n$ possible values of $k$. MATLAB code that implements this algorithm is available for download from \cite{code}, which is able to compute $C_{\tr}(\ketbra{x}{x})$ for pure states $\ket{x} \in \mathbb{C}^{1,000,000}$ in under one second on a standard laptop computer. We contrast this with the naive semidefinite program for computing $C_{\tr}(\ketbra{x}{x})$ \cite{RPL15}, which can only reasonably handle states in $\mathbb{C}^{100}$ or so.

Theorem~\ref{thm:pure_algorithm} can also be used to analytically compute $C_{\tr}(\ketbra{x}{x})$ for arbitrary pure states as well, as we now demonstrate with some examples.

\begin{example}	
	As a simple example, consider the qutrit pure state $\ket{x} = (2/3,2/3,1/3)$, which was investigated in~\cite{RPL15}. A direct calculation reveals that
	\bes
	q_1 & = & \frac{1}{6}\big(3\sqrt{5} - 5\big) \approx 0.2847, \\
	q_2 & = & \frac{1}{48}\big(3\sqrt{17} + 5\big) \approx 0.3619, \quad \text{and} \\
	q_3 & = & \frac{16}{45} \approx 0.3556.
	\ees
	
	\noindent Thus $k = 2$ (since $q_1 < x_1$ and $q_2 < x_2$, but $q_3 \geq x_3$), which then gives $C_{\tr}(\ketbra{x}{x}) = \frac{1}{6}\big(3 + \sqrt{17}\big)$ and $D = \mathrm{diag}(1/2,1/2,0)$, verifying that the state $D$ found in~\cite{RPL15} is indeed optimal.
\end{example}

\begin{example}	
	As another example, consider an arbitrary qubit pure state $\ket{x} = (x_1,x_2) \in \mathbb{C}^2$. Then
	\bes
	q_2 = \frac{|x_1x_2|}{|x_1|+|x_2|} \leq \mathrm{min}\{|x_1|,|x_2|\},
	\ees
	with equality if and only if either $x_1 = 0$ or $x_2 = 0$. If $x_1,x_2 \neq 0$ then $k = 2$ and we then have $C_{\tr}(\ketbra{x}{x}) = 2|x_1x_2|$ and $D = \mathrm{diag}(\ketbra{x}{x})$, which agrees with the formula for qubit states found in~\cite{RPL15}. If $x_1 = 0$ or $x_2 = 0$ then $k = 1$ and it is straightforward to check that we get the same formula.
\end{example}

\section{IV. Maximally coherent states under the trace norm of coherence}

We recall~\cite{BCP14} that a pure state $\ket{x} \in \mathbb{C}^n$ is called \emph{maximally coherent} if all of its entries have equal absolute value: $|x_1| = \cdots = |x_n| = 1/\sqrt{n}$. Recently it has been suggested that the maximum value of a proper measure of coherence should be attained exactly by the maximally coherent states \cite{PJF15}, and this property is known to hold for the relative entropy of coherence (this is straightforward to prove, see~\cite{BCP14} for example), the $\ell_1$-norm of coherence~\cite[Theorem 2]{SBDP15}, and the robustness of coherence~\cite{NBCPJA16}. We now show that this same property also holds for the trace distance of coherence, which provides further evidence that it is indeed a proper measure of coherence.

\begin{theorem}\label{thm:maxcoherent}
	For all (potentially mixed) states $\rho \in \cD_n$, we have $C_{\tr}(\rho) \leq 2 - 2/n$. 
	Furthermore, equality holds if and only if $\rho = \ketbra{x}{x}$, where $\ket{x}$ is a maximally coherent state.
\end{theorem}

We note that, while the upper bound in \ref{thm:maxcoherent} is well-known (see \cite[Theorem 2.1]{BDC}), the ``iff'' statement for equality was not. 

\begin{proof}
Let $\rho$ be a general mixed state with spectral decomposition 
$\sum_{j=1}^n p_j \ketbra{x_j}{x_j}$ such that 
$p_1 \ge \cdots \ge p_k > 1/n \ge p_{k+1} \ge \cdots \ge p_n$.
Then 
\begin{eqnarray*}
	\min_{\delta\in \cI} \|\rho - \delta\|_{\tr}  
	&\le&  \|\rho - I/n\|_{\tr} \\  
	&=& \sum_{j=1}^k (p_j - 1/n) +  \sum_{j=k+1}^n  (1/n-p_j)   \\
	&=& 2\sum_{j=1}^k (p_j - 1/n)   \\ 
	&\le& 2(1 - k/n)\\
	&\le& 2(1-1/n),
\end{eqnarray*}
where the second equality holds because $\tr(\rho - I/n) = 0$. If the equality 
$\min_{\delta\in \cI}\|\rho-\delta\|_{\tr} = 2(1-1/n)$ holds, then 
$k = 1$ so that $\rho =  \ketbra{x}{x}$ has rank one,
and $D = I/n$ satisfies
$C_{\tr}( \ketbra{x}{x}) = \| \ketbra{x}{x}-D\|_{\tr}$.
We may replace $\ket{x}$ by $PU\ket{x}$ as in 
Section~II and so we assume without loss of generality that 
$\ket{x} = (x_1, \dots, x_n)^t$ 
with $x_1 \ge \cdots \ge x_n \ge 0$. 
By Corollary 1 (2), we see that $d_1 = \cdots = d_n$ so that 
$s_n - nx_1 = \cdots = s_n - nx_n$. Thus, $x_1 = \cdots = x_n$.
The desired conclusion follows. \end{proof}

\section{V. Coherence and Entanglement Measures}

In this section, we show that a wide variety of measures of coherence coincide exactly with an analogous measure of entanglement when restricted to pure states. The motivating example for this result is that it (in conjunction with Theorem~\ref{thm:pure_algorithm}) gives an ``almost-formula'' on pure states for what we call the \emph{trace distance of entanglement}:
\[
E_{\tr}(\rho) := \min_{\sigma\in \mathcal{S}}\|\rho-\sigma\|_{\tr},
\]
where $\mathcal{S}$ is the set of separable states in a bipartite Hilbert space.

Throughout this section, we suppose without loss of generality that the state $\ket{v} \in \mathbb{C}^n \otimes \mathbb{C}^n$ has Schmidt decomposition $\ket{v} = \sum_{j=1}^n \lambda_j \ket{j}\otimes\ket{j}$, which is justified by multiplying $\ket{v}$ by some local unitaries and noting that all quantities we consider are invariant under local unitaries.
\begin{theorem}\label{thm:max_correlated_entanglement}
	Let $\ket{v} \in \mathbb{C}^n \otimes \mathbb{C}^n$ be a pure state with Schmidt coefficients $\lambda_1, \lambda_2, \ldots, \lambda_n$, and define $\ket{\lambda} := (\lambda_1,\lambda_2,\ldots,\lambda_n)$ to be the vector containing those Schmidt coefficients. Then $E_{\tr}(\ketbra{v}{v}) = C_{\tr}(\ketbra{\lambda}{\lambda})$.
\end{theorem}

We note that Theorem~\ref{thm:max_correlated_entanglement} is rather remarkable for the fact that it shows that computing $E_{\tr}(\ketbra{v}{v})$ is roughly as difficult as computing the Schmidt coefficients of $\ket{v}$ (and in particular, is thus computable in polynomial time). This was not obvious a priori, as optimizations over the set of separable states are typically NP-hard \cite{Gur03}, and in practice they are usually approximated by semidefinite programs that make use of symmetric extensions \cite{DPS04}.

Before proving this theorem, we present the lemma that is at its heart and does most of the heavy lifting.

\begin{lemma}\label{lem:ppt_incoherent}
	Let $\sigma \in \mathcal{M}_n \otimes \mathcal{M}_n$ be a real-valued state with positive partial transpose. Then there exists a quantum channel $\Phi : \mathcal{M}_n \otimes \mathcal{M}_n \rightarrow \mathcal{M}_n$ (which depends on $\sigma$) such that $\Phi(\sigma)$ is incoherent (i.e., diagonal), and $\Phi(\ketbra{v}{v}) = \ketbra{\lambda}{\lambda}$ for all pure states $\ket{v} \in \mathbb{C}^n \otimes \mathbb{C}^n$ of the form $\ket{v} = \sum_{j=1}^n\lambda_j\ket{j}\otimes\ket{j}$ (and $\ket{\lambda}$ is as defined in Theorem~\ref{thm:max_correlated_entanglement}).
\end{lemma}

\begin{proof}
	The channel $\Phi$ that works will be constructed as the composition of two simpler channels. To begin, we consider the diagonal twirling channel $\Psi : \mathcal{M}_n \otimes \mathcal{M}_n \rightarrow \mathcal{M}_n \otimes \mathcal{M}_n$ defined by
	$$
	\Psi(\rho) = \int_{U \in \mathcal{D}(\mathcal{U})} (U \otimes \overline{U}) \rho (U \otimes \overline{U})^\dagger \, dU,
	$$
	where $\mathcal{D}(\mathcal{U})$ is the set of diagonal unitary matrices, and we integrate with respect to the usual Haar measure. Then if $\rho_{ij,k\ell}$ denotes the coefficient of the basis matrix $\ketbra{i}{j}\otimes\ketbra{k}{\ell}$ in a density matrix $\rho$, we have the following:
	
	\noindent \textbf{Claim 1:} $$\Psi(\rho) = \sum_{i,j=1}^n \rho_{ii,jj}\ketbra{i}{i} \otimes \ketbra{j}{j} + \sum_{i\neq j=1}^n \rho_{ij,ij}\ketbra{i}{j} \otimes \ketbra{i}{j},$$
	
	\noindent and
	
	\noindent \textbf{Claim 2:} If $\sigma$ has positive partial transpose then so does $\Psi(\sigma)$.
	
	Claim~2 follows simply from the fact that conjugation by each $U \otimes \overline{U}$ does not change whether or not a state has positive partial transpose, so integrating over these states also gives a PPT state by convexity. 
	
	To see why Claim~1 holds, we explicitly compute the coefficient $\psi_{ij,k\ell}$ of $\ketbra{i}{j} \otimes \ketbra{k}{\ell}$ in $\Psi(\rho)$:
	\begin{eqnarray*}
		\psi_{ij,k\ell} &=& \iiiint_{U(1)} z_iz_\ell\overline{z_jz_k}\rho_{ij,k\ell} \, dz_i \, dz_j \, dz_k \, dz_k \\
		& = & \begin{cases}
			\rho_{ij,k\ell}, \quad \text{if} \quad (i,\ell) = (j,k) \\
			\rho_{ij,k\ell}, \quad \text{if} \quad (i,\ell) = (k,j) \\
			0, \quad \text{otherwise}
		\end{cases},
	\end{eqnarray*}
	where the first two cases follow simply from the fact that $z_iz_\ell\overline{z_jz_k} = |z_i|^2|z_\ell|^2 = 1$ if $(i,\ell) = (j,k)$ or $(i,\ell) = (k,j)$, and the third case follows from invariance of the Haar measure and the fact that $\int_{U(1)} z_i \, dz_i = 0$.
	
	Before proceeding, we note that Claim~1 implies in particular that $\Psi(\ketbra{v}{v}) = \ketbra{v}{v}$ (recall that we are assuming that $\ket{v} = \sum_{j=1}^n \lambda_j \ket{j}\otimes\ket{j}$), and Claim~2 implies that the matrix $$\sum_{i,j=1}^n \sigma_{ii,jj}\ketbra{i}{i} \otimes \ketbra{j}{j} + \sum_{i\neq j=1}^n \sigma_{ij,ij}\ketbra{i}{j} \otimes \ketbra{i}{j}$$ has positive partial transpose. By computing the partial transpose of this matrix, we thus see that every $2 \times 2$ matrix of the form $$\begin{bmatrix}\sigma_{ii,jj} & \sigma_{ij,ij} \\ \sigma_{ji,ji} & \sigma_{jj,ii}\end{bmatrix}$$ must be positive semidefinite. We thus conclude that $$|\sigma_{ij,ij}| \leq \sqrt{\sigma_{ii,jj}\sigma_{jj,ii}} \leq \frac{1}{2}(\sigma_{ii,jj} + \sigma_{jj,ii})$$ for all $i \neq j$. We define $c := \sqrt{\frac{2|\sigma_{ij,ij}|}{\sigma_{ii,jj} + \sigma_{jj,ii}}}$, which is thus a real number between $0$ and $1$, and $s := \sigma_{ij,ij}/|\sigma_{ij,ij}|$, which is the sign of $\sigma_{ij,ij}$ (recall that we are assuming $\sigma$ is real-valued, so $\sigma_{ij,ij}$ is a real number).
	
	Now that we have established all of the properties of $\Psi$ that we need, we introduce one more channel that will be used to finish the proof. This channel, which we denote by $\Omega_\sigma : \mathcal{M}_n \otimes \mathcal{M}_n \rightarrow \mathcal{M}_n$, depends on $\sigma$ and is defined via the following set of $1 + 2n(n-1)$ Kraus operators:
	\begin{eqnarray*}
		E_+ & := & \sum_{j=1}^n \ket{j}(\bra{j} \otimes \bra{j}) \\
		E_{ij} & := & \frac{c}{\sqrt{2}}(\ket{i} - s\ket{j})(\bra{i} \otimes \bra{j}) \quad \text{ for all } 1 \leq i \neq j \leq n \\
		F_{ij} & := & \sqrt{1-c^2}\ket{i}(\bra{i} \otimes \bra{j}) \quad \text{ for all } 1 \leq i \neq j \leq n.
	\end{eqnarray*}
	
	To see that $\Omega_\sigma(\ketbra{v}{v}) = \ketbra{\lambda}{\lambda}$, we compute
	\begin{eqnarray*}
		&& \Omega_\sigma(\ketbra{v}{v}) \\ 
		& = & E_+\ketbra{v}{v} E_+^\dagger + 
		\sum_{i\neq j} E_{ij}\ketbra{v}{v} E_{ij}^\dagger 
		+ \sum_{i\neq j} F_{ij}\ketbra{v}{v} F_{ij}^\dagger \\
		& = & \sum_{i,j=1}^n \lambda_i\lambda_j\ketbra{i}{j} + 0 + 0 \\
		& = & \ketbra{\lambda}{\lambda}.
	\end{eqnarray*}
	
	To see that $\Omega_\sigma(\Psi(\sigma))$ is incoherent, we verify that
	\begin{eqnarray*}
	&&	\Omega_\sigma(\Psi(\sigma)) \\ 
		& = & E_+\Psi(\sigma) E_+^\dagger + \sum_{i\neq j} E_{ij}\Psi(\sigma) E_{ij}^\dagger 
		 + \sum_{i\neq j} F_{ij}\Psi(\sigma) F_{ij}^\dagger \\
		& = & \sum_{i,j=1}^n \sigma_{ij,ij}\ketbra{i}{j} 
		 + (1-c^2)\sum_{i\neq j=1}^n \sigma_{ii,jj}\ketbra{i}{i} \\&& 
		+ \frac{c^2}{2}\sum_{i\neq j=1}^n \sigma_{ii,jj}(\ket{i} - s\ket{j})(\bra{i} - s\bra{j}) 
		\\
		& = & \sum_{i=1}^n \left(\sigma_{ii,ii} 
		+ (c^2 + (1-c^2))\sum_{j\neq i} \sigma_{ii,jj}\right)\ketbra{i}{i} \\
		&& + \sum_{i=1 \neq j}^n \left(\sigma_{ij,ij} - \frac{sc^2}{2}(\sigma_{ii,jj} + \sigma_{jj,ii})\right)\ketbra{i}{j}, \\
		& = & \sum_{i=1}^n \left(\sum_{j=1}^n \sigma_{ii,jj}\right)\ketbra{i}{i},
	\end{eqnarray*}
	which is incoherent.
	
	Finally, we must verify that $\Omega_\sigma$ is a quantum channel. It is completely positive by construction (any map defined in terms of Kraus operators is), so we just need to verify that it is trace-preserving (i.e., the fact that $\Omega_\sigma^\dagger(I) = I$). To this end, we compute
	\begin{eqnarray*}
		\Omega_\sigma^\dagger(I) & = & E_+^\dagger E_+ + \sum_{i\neq j} E_{ij}^\dagger E_{ij} + \sum_{i\neq j} F_{ij}^\dagger F_{ij} \\
		&=& \sum_{i=1}^n \ketbra{i}{i} \otimes \ketbra{i}{i} \\
		&& + (c^2 + (1-c^2))\sum_{i\neq j=1}^n \ketbra{i}{i} \otimes \ketbra{j}{j} \\
		& = & I.
	\end{eqnarray*}
	
	We have thus shown that $\Omega_\sigma(\Psi(\sigma))$ is incoherent, and $\Omega_\sigma(\Psi(\ketbra{v}{v})) = \Omega_\sigma(\ketbra{v}{v}) = \ketbra{\lambda}{\lambda}$ for all pure states of the form $\ket{v} = \sum_{j=1}^n\lambda_j\ket{j}\otimes\ket{j}$, so the channel $\Phi := \Omega_\sigma \circ \Psi$ is the one described by the lemma.
\end{proof}

\begin{proof}[Proof of Theorem~\ref{thm:max_correlated_entanglement}]
	We start by proving that $E_{\tr}(\ketbra{v}{v}) \leq C_{\tr}(\ketbra{\lambda}{\lambda})$. To this end, let $\delta^* = \mathrm{diag}(\delta_1^*,\ldots,\delta_n^*) \in \mathcal{I}$ be an incoherent state that attains the minimum in $C_{\tr}(\ketbra{\lambda}{\lambda})$: $$C_{\tr}(\ketbra{\lambda}{\lambda}) = \min_{\delta\in \mathcal{I}}\|\ketbra{\lambda}{\lambda}-\delta\|_{\tr} =\|\ketbra{\lambda}{\lambda}-\delta^*\|_{\tr}.$$ Then consider the separable state	
	$$\sigma^* = \sum_{i=1}^n \delta_i^* \ket{i}\bra{i} \otimes \ket{i}\bra{i}.$$
	
	A calculation then reveals that
\begin{eqnarray*}
		&&E_{\tr}(\ketbra{v}{v}) \\ 
		&=& \min_{\sigma \in \mathcal{S}} \|\ketbra{v}{v}-\sigma\|_{\tr}\\
		&\leq& \|\ketbra{v}{v}-\sigma^*\|_{\tr} \\
		&=& \left\|\sum_{i,j=1}^n \lambda_i\lambda_j \ket{i}\bra{j} \otimes \ket{i}\bra{j}-\sum_{i=1}^n \delta_i^* \ket{i}\bra{i} \otimes \ket{i}\bra{i}\right\|_{\tr}.
		\end{eqnarray*}
	
	We recognize that the matrix on the far right above is exactly the same as the matrix $$\sum_{i,j=1}^n \lambda_i\lambda_j\ket{i}\bra{j}-\sum_{i=1}^n \delta_i^* \ket{i}\bra{i} = \tilde{\rho} - \delta^*,$$
	but with some extra rows and columns of zeroes. Since those rows and columns of zeroes do not affect the trace norm, it follows that $E_{\tr}(\ketbra{v}{v}) \leq \|\ketbra{\lambda}{\lambda} - \delta^*\|_{\tr} = C_{\tr}(\ketbra{\lambda}{\lambda})$, as desired.
	
	Next, we prove the inequality that $C_{\tr}(\ketbra{\lambda}{\lambda}) \leq E_{\tr}(\ketbra{v}{v})$ in a very similar manner. To this end, let $\sigma^* \in \mathcal{S}$ be a separable (and hence PPT) state that attains the minimum in $E_{\tr}(\ketbra{v}{v})$: $$E_{\tr}(\ketbra{v}{v}) = \min_{\sigma\in \mathcal{S}}\|\ketbra{v}{v}-\sigma\|_{\tr} =\|\ketbra{v}{v}-\sigma^*\|_{\tr}.$$ Note that we can assume without loss of generality that $\sigma^*$ has all real entries, since $$\big\|\ketbra{v}{v} - \frac{1}{2}(\sigma^* + (\sigma^*)^T)\big\|_{\tr} \leq \|\ketbra{v}{v} - \sigma^*\|_{\tr}$$ by the triangle inequality, and the separable state $\frac{1}{2}(\sigma^* + (\sigma^*)^T)$ has all real entries.
	
	Then let $\Phi : \mathcal{M}_n \otimes \mathcal{M}_n \rightarrow \mathcal{M}_n$ be the channel described by Lemma~\ref{lem:ppt_incoherent} and let $\delta^* := \Phi(\sigma^*)$ (which is an incoherent state). Observe that
	\begin{eqnarray*}
		C_{\tr}(\ketbra{\lambda}{\lambda}) & = & \min_{\delta \in \mathcal{I}} \|\ketbra{\lambda}{\lambda}-\delta\|_{\tr} \\
		& \leq & \|\ketbra{\lambda}{\lambda}-\delta^*\|_{\tr} \\
		& = & \|\Phi(\ketbra{v}{v} -\sigma^*)\|_{\tr} \\
		& \leq & \|\ketbra{v}{v} -\sigma^*\|_{\tr} \\
		& = & E_{\tr}(\ketbra{v}{v}),
	\end{eqnarray*}
	where the final inequality comes from the fact that $\|\Phi\|_{\diamond} \leq 1$ for all quantum channels $\Phi$, and thus $\Phi$ cannot increase the trace norm. This completes the proof.
\end{proof}

The proof of one of the inequalities in Theorem~\ref{thm:max_correlated_entanglement} was quite straightforward, while the other inequality required the use of Lemma~\ref{lem:ppt_incoherent}. The same technique can be used to prove that other measures of entanglement and coherence coincide on pure states as well. For example, for the \emph{robustness of entanglement} \cite{VT99} $R_E$, we could use this method to show that $R_E(\ketbra{v}{v}) = C_R(\ketbra{\lambda}{\lambda})$ (however explicit formulas are already known for each of $R_E(\ketbra{v}{v})$ and $C_R(\ketbra{\lambda}{\lambda})$, so this does not get us anything new).

However, we also note that Theorem~\ref{thm:max_correlated_entanglement} and Lemma~\ref{lem:ppt_incoherent} can both be generalized slightly from pure states to real-valued states that are \emph{maximally correlated} \cite{Rai01}: states with the special form $\rho = \sum_{i,j=1}^n \rho_{ij}\ketbra{i}{j} \otimes \ketbra{i}{j}$, which lets us show that coherence measures and entanglement measures also coincide on this slightly wider class of states as well (rather than just on pure states).

\section{VI. Relationship between the $\ell_1$-norm of coherence and the relative entropy of coherence}

Consider Proposition~5 of \cite{RPL15}, which asserts that the $\ell_1$-norm coherence of a 
pure state is never smaller than its relative entropy of coherence. This section is devoted to 
providing an alternate proof to this theorem. In particular, the authors in \cite{RPL15} use the 
recursive property of the entropy function to show that $C_{\ell_1}\geq C_r$ for all pure states. 
Our proof, on the other hand, relies on showing that a function remains non-negative upon small 
perturbations of the components of its input. Much detail is given, with the hope of better 
understanding this inequality.

Before proceeding, recall that the relative entropy of coherence is defined in terms of the von Neumann entropy $S(\rho) := -\tr(\rho\log_2(\rho))$. From now on, we will write $\log=\log_2$ for notational simplicity, since we deal with no other base. 

\begin{theorem}\label{thm:conj6}
	Suppose $\{\lambda_i\}_{i=1}^n$ are such that $\sum\limits_i \lambda_i=1$ and $\lambda_i\geq 0$ for every $i$. Then
	\begin{eqnarray*}
		-\sum\limits_i \lambda_i \log \lambda_i \leq \Big(\sum\limits_i \sqrt{\lambda_i}\Big)^2-1.
	\end{eqnarray*}
\end{theorem}

\begin{proof}
	In order to prove the above inequality, it suffices to show that the function
	$f(\vec{\lambda}):=(\sum\limits_i \sqrt{\lambda_i})^2-1+\sum\limits_i \lambda_i \log \lambda_i$ is always non-negative for any probability vector $\vec{\lambda}$.
	
	Without loss of generality, we can assume all $\lambda_i$'s are strictly positive. Otherwise, we just look at some smaller $n$.  Let's consider the following perturbation:
	\begin{eqnarray*}
		&&f(\{\lambda_1,\cdots,\lambda_{i-1}, \lambda_i-\epsilon, 
		\lambda_{i+1},\cdots, \lambda_{j-1},\lambda_j+\epsilon,\cdots, \lambda_n\})
		\\ && -f(\{\lambda_1,\cdots,\lambda_n\})\\
		&&= \big(\sum\limits_{k\neq i, j} \sqrt{\lambda_k}+\sqrt{\lambda_i-\epsilon}+\sqrt{\lambda_j+\epsilon}\big)^2
		\\ && -\big(\sum\limits_{k\neq i,j} \sqrt{\lambda_k}+\sqrt{\lambda_i}+\sqrt{\lambda_j}\big)^2\\		&&+(\lambda_i-\epsilon)\log(\lambda_i-\epsilon)+(\lambda_j+\epsilon)\log(\lambda_j+\epsilon)
		\\ && -\lambda_i \log \lambda_i-\lambda_j \log \lambda_j\\
		&&=\big(2\sum\limits_{k\neq i, j} \sqrt{\lambda_k}+\sqrt{\lambda_i-\epsilon}+\sqrt{\lambda_j+\epsilon}
		+\sqrt{\lambda_i}+\sqrt{\lambda_j}\big) \\ &&\times (\sqrt{\lambda_i-\epsilon}+\sqrt{\lambda_j+\epsilon}-\sqrt{\lambda_i}-\sqrt{\lambda_j})\\
		&&+\lambda_i \log\big(1-\frac{\epsilon}{\lambda_i}\big)+\lambda_j\log\big(1+\frac{\epsilon}{\lambda_j}\big) 
		\\ && +\epsilon[\log(\lambda_j+\epsilon)-\log(\lambda_i-\epsilon)].
	\end{eqnarray*}
Recall that $\sqrt{1+x}=1+\frac{1}{2}x-\frac{1}{8}x^2+O(x^3)$ 
and $\log{(1+x)}=x-\frac{x^2}{2}+O(x^3)$, the above expression simplifies as 
	\begin{eqnarray*}
		&&\left[2\sum\limits_{k\neq i, j} \sqrt{\lambda_k}+\sqrt{\lambda_i}\big(2-\frac{\epsilon}{2\lambda_i}\big)+\sqrt{\lambda_j}\big(2+\frac{\epsilon}{2\lambda_j}\big)\right] \\
		&& \times\left(-\frac{\sqrt{\lambda_i} \epsilon}{2\lambda_i}+ \frac{\sqrt{\lambda_j}\epsilon}{2\lambda_j}\right)+\epsilon(\log \lambda_j-\log \lambda_i)+O(\epsilon^2) \\
		&=&  \bigg(\big(\sum\limits_{k=1}^n \sqrt{\lambda_k}\big)\left(\frac{1}{\sqrt{\lambda_j}}-\frac{1}{\sqrt{\lambda_i}}\right)+(\log \lambda_j-\log \lambda_i)\bigg)\epsilon\\
		&&+O(\epsilon^2).
	\end{eqnarray*}
	
	So, if  $\big(\sum\limits_{k=1}^n \sqrt{\lambda_k}\big)\left(\frac{1}{\sqrt{\lambda_j}}-\frac{1}{\sqrt{\lambda_i}}\right)+(\log \lambda_j-\log \lambda_i)<0$, then the above perturbation will lead to a smaller value of $f$. Thus, if we assume function $f$ achieves its minimum value at point $(\lambda_1,\cdots,\lambda_n)$, if $\lambda_i=0$ for some $i$, then we can look at the same problem with $n-1$ variables. So without loss of generality, we can still assume all $\lambda_i$'s are strictly positive, we must have 
	\begin{eqnarray*}
	\big(\sum\limits_{k=1}^n \sqrt{\lambda_k}\big)\left(\frac{1}{\sqrt{\lambda_j}}-\frac{1}{\sqrt{\lambda_i}}\right)+(\log \lambda_j-\log \lambda_i)\geq 0
	\end{eqnarray*}
	for any $1\leq i,j\leq n$.
	
	We can also consider the perturbation $(\lambda_1,\cdots,\lambda_n)\mapsto (\lambda_1,\cdots,\lambda_{i-1}, \lambda_i+\epsilon, 
		\lambda_{i+1},\cdots, \lambda_{j-1},\lambda_j-\epsilon,\lambda_{j+1},\cdots, \lambda_n)$ which will imply 
	\begin{eqnarray*}
	\big(\sum\limits_{k=1}^n \sqrt{\lambda_k}\big)\left(\frac{1}{\sqrt{\lambda_i}}-\frac{1}{\sqrt{\lambda_j}}\right)+(\log \lambda_i-\log \lambda_j)\geq 0
	\end{eqnarray*}
	for any $1\leq i,j\leq n$.
	
	By combining the above inequalities together, we will have
	\begin{eqnarray*}
	\big(\sum\limits_{k=1}^n \sqrt{\lambda_k}\big)\left(\frac{1}{\sqrt{\lambda_j}}-\frac{1}{\sqrt{\lambda_i}}\right)+(\log \lambda_j-\log \lambda_i)=0
	\end{eqnarray*}
	for any $1\leq i,j\leq n$.
	
	It also implies that, $\frac{\log \lambda_j-\log \lambda_i}{\frac{1}{\sqrt{\lambda_j}}-\frac{1}{\sqrt{\lambda_i}}}=-\sum\limits_{k=1}^n \sqrt{\lambda_k}$ for any $1\leq i\neq j\leq n$. 
	
	Note that, for any given $0\leq t\leq 1$,  function $g_t(x)=\frac{\log x-\log t}{\frac{1}{\sqrt{x}}-\frac{1}{\sqrt{t}}}$ is a decreasing function for $x\in (0,1]$. Hence, if $n\geq 3$ and there are at least two distinct $\lambda_i$ and $\lambda_{i'}$,  let's choose $j\neq i, i'$, we must have $g_{\lambda_j}(\lambda_i)\neq g_{\lambda_j}(\lambda_i')$. It's a contradiction. Thus, we must have $n\leq 2$ or $\lambda_1=\lambda_2=\cdots=\lambda_n$ in which case $f(\vec{\lambda})=n-1+\log \frac{1}{n}=n-1- \log n$, which is always non-negative for $n\in \mathbb{Z}^{+}$. For the case $n\leq 2$, we have $\lambda_1$ and $\lambda_2=1-\lambda_1$ satisfying the minimum condition:
	\begin{eqnarray*}
	&& (\sqrt{\lambda_1}+\sqrt{1-\lambda_1})(\frac{1}{\sqrt{\lambda_1}}-\frac{1}{\sqrt{1-\lambda_1}})\\
	&&+\log(\lambda_1)-\log(1-\lambda_1)=0\\
	&\iff & \sqrt{\frac{1-\lambda_1}{\lambda_1}}-\sqrt{\frac{\lambda_1}{1-\lambda_1}}+\log \frac{\lambda_1}{1-\lambda_1}=0.
 	\end{eqnarray*}
	
	We also have
	\begin{eqnarray*}
	&&f(\{\lambda_1,1-\lambda_1\})\\
	&=&(\sqrt{\lambda_1}+\sqrt{1-\lambda_1})^2-1+\lambda_1 \log \lambda_1+(1-\lambda_1)\log(1-\lambda_1)\\
	&=&2\sqrt{\lambda_1(1-\lambda_1)}+\lambda_1 \log \lambda_1+(1-\lambda_1)\log(1-\lambda_1)\\
	&=& 2\sqrt{\lambda_1(1-\lambda_1)}+\log(1-\lambda_1)+\lambda_1(\log\lambda_1-\log(1-\lambda_1))\\
	&=&2\sqrt{\lambda_1(1-\lambda_1)}+\log(1-\lambda_1)-\lambda_1(\sqrt{\frac{1-\lambda_1}{\lambda_1}}-\sqrt{\frac{\lambda_1}{1-\lambda_1}})\\
	&=&\sqrt{\frac{\lambda_1}{1-\lambda_1}}+\log(1-\lambda_1).
	\end{eqnarray*}
	
	Let $t=\sqrt{\frac{\lambda_1}{1-\lambda_1}}$, our aim is to show $t-\log(1+t^2)\geq 0$ under the assumption that $2\log t+\frac{1}{t}-t=0$. It is easy to verify that it has three roots only: $t_1=0.215106$, $t_2=1$ and $t_3=4.64886$, and for all of them $t-\log(1+t^2)\geq 0$. The result follows. 
\end{proof}

\begin{corollary}\label{cor:l1_relent}
	For every pure state $\ket{x}$, 
	\begin{eqnarray*}
		C_{\ell_1}(\ketbra{x}{x})\geq \max\{C_r(\ketbra{x}{x}), 2^{C_r(\ketbra{x}{x})}-1\}.
	\end{eqnarray*}
\end{corollary}

\begin{proof} Write $\ket{x}=\sum_{i=1}^n\sqrt{\lambda_i}\ket{i}$ for a given basis $\{\ket{i}\}_{i=1}^n$. Then $C_{\ell_1}(\ketbra{x}{x})=\big(\sum_{i=1}^n\sqrt{\lambda_i}\big)^2-1$. Recall that the von Neumann entropy is zero for pure states, and so $C_r(\ketbra{x}{x})$ reduces to $S(\ketbra{x}{x}_{\operatorname{diag}})=-\sum_{i=1}^n\lambda_i\log\lambda_i$. 
	In \cite[Proposition 5]{RPL15}, the authors prove that $C_{\ell_1}(\ketbra{x}{x})\geq 2^{C_r(\ketbra{x}{x})}-1$. Theorem \ref{thm:conj6} above states that $C_{\ell_1}(\ketbra{x}{x})\geq  -\sum_{i=1}^n\lambda_i\log\lambda_i= C_r(\ketbra{x}{x})$.
\end{proof}

We note that Corollary~\ref{cor:l1_relent} improves the bound $C_{\ell_1}(\ketbra{x}{x})\geq \ln(2)C_r(\ketbra{x}{x})$ given in \cite[Proposition 5]{RPL15}. We also note, following the discussion in \cite[Section~III]{RPL15}, that Corollary~\ref{cor:l1_relent} improves a well-known inequality relating the negativity~\cite{VW02} and the distillable entanglement of pure states.

To elaborate a bit, we recall that a well-known upper bound on the relative entropy of entanglement $E_r(\ketbra{y}{y})$ of a pure state $\ket{y} \in \mathbb{C}^m \otimes \mathbb{C}^n$ in terms of its negativity $N(\ketbra{y}{y})$ is $E_r(\ketbra{y}{y}) \leq \log(1 + 2N(\ketbra{y}{y}))$. Since the relative entropy of entanglement is equal to the distillable entanglement when restricted to pure states, the same inequality holds for distillable entanglement as well. Using our results, we immediately obtain the following improvement to this bound:

\begin{corollary}\label{cor:entanglement}
	For every pure state $\ket{y} \in \mathbb{C}^m \otimes \mathbb{C}^n$, 
	\begin{eqnarray*}
		E_r(\ketbra{y}{y}) \leq 2N(\ketbra{y}{y}).
	\end{eqnarray*}
\end{corollary}

\begin{proof}
	It is straightforward to verify that if $\ket{x} = \sum_{i=1}^n x_i\ket{i}$ and $\ket{y}$ has Schmidt coefficients $\{x_i\}_{i=1}^n$, then $C_{\ell_1}(\ketbra{x}{x}) = 2N(\ketbra{y}{y})$ and $C_r(\ketbra{x}{x}) = E_r(\ketbra{y}{y})$. Thus using Corollary~\ref{cor:l1_relent} immediately implies that
	\bes
	E_r(\ketbra{y}{y}) &=& C_r(\ketbra{x}{x}) \leq C_{\ell_1}(\ketbra{x}{x}) = 2N(\ketbra{y}{y}),
	\ees
	as desired.
\end{proof}

It is straightforward to verify that the bound provided by Corollary~\ref{cor:entanglement} is strictly better than the known bound $E_r(\ketbra{y}{y}) \leq \log(1 + 2N(\ketbra{y}{y}))$ exactly when $N(\ketbra{y}{y}) < 1/2$.

\section{VII. Conclusions and Discussion}

In this work, we derived an explicit expression for the trace distance of coherence of a pure state, as well as the closest incoherent state to a given pure state with respect to the trace distance. One natural question that arises from this work is whether or not Theorem~\ref{thm:pure_algorithm} can be used to show that the trace distance of coherence is strongly monotonic under incoherent quantum channels (and is thus a proper coherence measure), at least when it is restricted to pure states. We also proved that the states maximizing the trace distance of coherence are exactly the maximally coherent states, which provides evidence in favor of it being a proper coherence measure.

We gave an alternate proof to the recent theorem that the $\ell_1$-norm of coherence is not smaller than the relative entropy of coherence for pure states (Corollary~\ref{cor:l1_relent}), and showed how  this result is used to derive a new relationship between negativity and distillable entanglement of pure states. We note that it has been conjectured that the same relationship between the $\ell_1$-norm of coherence and the relative entropy of coherence holds even for arbitrary mixed states. This   conjecture is beyond the scope of our work; our perturbation techniques for the case of pure states rely on the linearity of the first-order term, which is no longer linear for the mixed state case. Perturbation techniques may still apply if we study higher-order terms, however, more detailed calculation may be involved.

In a further attempt to analyze the trace measure of coherence, we show that it is precisely the same quantity as the analogous trace distance of entanglement when restricted to pure states. In particular, this gives an efficient method of computing the trace distance of entanglement for pure states, and it was not obvious a priori that such a method even existed. More generally, we showed that many natural pairs of coherence and entanglement measures share the exact same formulas when restricted to pure states: the entanglement in a pure state is equal to the coherence of its vector of Schmidt coefficients, and this property generalizes slightly to the class of real maximally correlated states.

\section{Acknowledgements}

The authors would like to thank Diana Pelejo and Yue Liu for some helpful comments in the early
discussion of the project. They also thank Marco Piani for useful discussions about an early draft of this paper, and for conjecturing some results that led to the entirety of Section~V. S.G. was supported by the Mount Allison President's Research and Creative Activities Fund. C.-K.L. is an affiliate member of the Institute for Quantum Computing, 
University of Waterloo. He is an honorary professor of the University of Hong Kong and the Shanghai 
University. His research was supported
by USA NSF grant DMS 1331021, Simons Foundation Grant 351047, and NNSF of China
Grant 11571220. S.P. was supported by NSERC Discovery Grant number 1174582.

\end{document}